\PassOptionsToPackage{ngerman,english}{babel}
\documentclass[a4paper,cleveref,autoref,thm-restate]{lipics-v2021}
\nolinenumbers

\title{Weighted HOM-Problem for Nonnegative Integers}
\titlerunning{Weighted HOM-Problem for Nonnegative Integers}

\author{Andreas Maletti}{Institute of Computer Science, Leipzig University, 04109 Leipzig, Germany}{andreas.maletti@uni-leipzig.de}{https://orcid.org/0000-0002-0814-598X}{}

\author{Andreea-Teodora N{\'{a}}sz}{Institute of Computer Science, Leipzig University, 04109 Leipzig, Germany}{nasz@informatik.uni-leipzig.de}{https://orcid.org/0000-0002-0814-598X}{Research financially supported by a scholarship awarded to T.~Nasz by the Free State of Saxony.}{}

\author{Erik Paul}{Institute of Computer Science, Leipzig University, 04109 Leipzig, Germany}{epaul@informatik.uni-leipzig.de}{https://orcid.org/0000-0002-0814-598X}{}

\authorrunning{A. Maletti, T. Nasz and E. Paul} 

\Copyright{Andreas Maletti, Andreea-Teodora N{\'{a}}sz and Erik Paul} 



\category{} 

\relatedversion{} 

\EventEditors{John Q. Open and Joan R. Access}
\EventNoEds{2}
\EventLongTitle{42nd Conference on Very Important Topics (CVIT 2016)}
\EventShortTitle{CVIT 2016}
\EventAcronym{CVIT}
\EventYear{2016}
\EventDate{December 24--27, 2016}
\EventLocation{Little Whinging, United Kingdom}
\EventLogo{}
\SeriesVolume{42}
\ArticleNo{23}

\usepackage[utf8]{inputenc}
\usepackage{thm-restate}
\usepackage{etoolbox}
\usepackage{amsmath,amsfonts,amssymb,dsfont,centernot}
\usepackage{stmaryrd,mathtools,tikz,forest,graphicx}
\usetikzlibrary{arrows,positioning,decorations.pathreplacing}
\usepackage{multirow,xcolor}
\usepackage[normalem]{ulem}
\useunder{\uline}{\ul}{}

\allowdisplaybreaks

\useshorthands{"}
\addto\extrasenglish{\languageshorthands{ngerman}}
\DeclareMathOperator{\wt}{wt}
\DeclareMathOperator{\pos}{pos}

\DeclareMathOperator{\var}{var}
\DeclareMathOperator{\rk}{rk}
\DeclareMathOperator{\he}{ht}
\DeclareMathOperator{\size}{size}
\providecommand*{\nat}[0]{\ensuremath{\mathbb N}}
\providecommand*{\seq}[3]{\ensuremath{#1_{#2}, \dotsc, #1_{#3}}}
\providecommand*{\abs}[1]{\ensuremath{\lvert #1 \rvert}}
\providecommand*{\word}[3]{\ensuremath{#1_{#2} \dotsm #1_{#3}}}

\newcommand{\nve}[1]{[#1]}
\newcommand{\nze}[2]{[#1,#2]}
\newcommand{\fixedn}{N}%

\newcommand{\Tsigma}{T_{\Sigma}}

\providecommand{\sem}[1]{\ensuremath{\llbracket #1 \rrbracket}}

\newcommand{\N}{\mathbb{N}}

\begin{document}

\maketitle

\begin{abstract}
  The HOM-problem asks whether the image of a regular tree language
  under a given tree homomorphism is again regular.  It was recently
  shown to be decidable by \textsc{Godoy}, \textsc{Gim{\'e}nez},
  \textsc{Ramos}, and \textsc{\`Alvarez}.  In this paper, the
  $\N$-weighted version of this problem is considered and its
  decidability is proved.  More precisely, it is decidable in
  polynomial time whether the image of a regular $\N$"~weighted tree
  language under a nondeleting, nonerasing tree homomorphism is
  regular.
\end{abstract}

\section{Introduction}
The prominent model of nondeterministic finite-state string automata 
has seen a variety of extensions in the past few decades.  Notably, their
qualitative evaluation was generalized to a quantitative one by means
of weighted automata in~\cite{schutzenberger1961}.  Those automata
have been extensively studied~\cite{salomaa2012automata}, not least
because of their ability to neatly represent process factors such as
costs, consumption of resources or time, and probabilities related to
the processed input.  Semirings~\cite{gol99,hebwei98} present
themselves as a well suited algebraic structure for the evaluation of
the weights because of their generality as well as their reasonable
computational efficiency that is derived from distributivity. 

Parallel to this development, finite-state automata have been generalized 
to process other forms of inputs such as infinite words~\cite{infinitewords} 
and trees~\cite{tataok}.  Finite-state tree automata and the \emph{regular 
tree languages} they generate have been widely researched since their
introduction in
\cite{doner1970tree,thatcher1965generalized,thatcher1968generalized}.  
These models have proven to be useful in a variety of application
areas including natural language processing~\cite{jurmar08}, image
generation~\cite{drewes2006grammatical}, and compiler
construction~\cite{wilseihac13}.  In many cases, applications require
the integration of both the quantitative evaluation and trees as a more
expressive input structure, which led to the development of several
weighted tree automaton~(WTA) models.  An extensive overview can be
found in~\cite[Chapter~9]{fulvog09}.

Finite-state tree automata have several serious limitations including
their inability to ensure the equality of two subtrees of any size in
an accepted tree.  These restrictions are well-known~\cite{gecste15},
and the mentioned drawback was addressed in~\cite{rateg1981}, where an
extension was proposed that is capable of explicitly requiring certain
subtrees to be equal or different.  These models are highly 
convenient in the study of tree transformations~\cite{fulvog09}, which
can implement subtree duplication, and they are also the primary tool
used in the seminal paper~\cite{godoy2013hom}, where the decidability
of the HOM-problem was established.

The HOM-problem, a previously long-standing open question in the study
of tree languages, asks whether the image of a regular tree language
under a given tree homomorphism is also regular.  The image need not
be regular since tree homomorphisms can generate copies of subtrees.
Indeed, if this copying ability is removed from the tree homomorphism
(e.g., linear tree homomorphisms), then the image is always
regular~\cite{gecste15}.  The classical (\textsc{Boolean}) HOM-problem
was  recently solved in~\cite{godoy2013hom,godoy2010hom}, where the
image is represented by a tree automaton with constraints, for which
it  is then determined whether it generates a regular tree
language. The problem was later shown to be
EXPTIME-complete~\cite{hom2012exp}.

In the weighted case, decidability of the HOM-problem remains open.
Previous research on the preservation of regularity in the weighted
setting~\cite{bozrah05,esikui03,fulmalvog10,fulmalvog10b} focuses on
cases that explicitly exclude the copying power of the homomorphism.
In the present work, we prove that the HOM-problem for regular
$\N$"~weighted tree languages can easily be decided in polynomial
time.  Our proof outline is inspired by~\cite{godoy2013hom}: Consider
a regular $\N$"~weighted tree language and a nondeleting, nonerasing
tree homomorphism. First, we represent this image efficiently using an
extension (WTGh) of weighted tree automata \cite{WTAc-journal}. The
question is now regularity of this WTGh, and the idea behind our
contribution is the reduction of its (non)regularity to something more
tangible: the large duplication property.
In turn, we prove decidability in polynomial time of the large
duplication property directly in Lemma~\ref{lm: ldpp decid}.  If the
WTGh for the homomorphic image does not have this property, then we
give an effective construction of an equivalent $\N$"~weighted WTG
(albeit in exponential time), thus proving its regularity.  Otherwise,
we use a pumping lemma presented in~\cite{WTAc-journal} and isolate a
strictly nonregular part from the WTGh.  The most challenging part of
our proof and our main technical contribution is showing that the
remaining part of the homomorphic image cannot compensate for this
nonregular behavior.  For this, we employ \textsc{Ramsey}'s
theorem~\cite{ramsey1930} to identify a witness for the nonregularity
of the whole weighted tree language.

\section{Preliminaries}
We denote the set of nonnegative integers by~$\N$.  For~$i,j \in \N$
we let~$\nze{i}{j} = \{k \in \N \mid i \leq k \leq j\}$ and~$\nve{j} =
\nze{1}{j}$.  Let~$Z$ be an arbitrary set.  The cardinality of~$Z$ is
denoted by~$\abs Z$, and the set of words over~$Z$ (i.e., the set of
ordered finite sequences of elements of~$Z$) is denoted by~$Z^*$.

\subsection*{Trees, Substitutions, and Contexts}
A \emph{ranked alphabet}~$(\Sigma, \mathord{\rk})$ consists of a
finite set~$\Sigma$ and a mapping~$\mathord{\rk} \colon \Sigma \to \N$
that assigns a rank to each symbol of~$\Sigma$.  If there is no risk 
of confusion, then we denote the ranked alphabet~$(\Sigma,
\mathord{\rk})$ by~$\Sigma$ alone.  We write~$\sigma^{(k)}$ to
indicate that~$\rk(\sigma) = k$.  Moreover, for every~$k \in \N$ we
let~$\Sigma_k = \rk^{-1}(k)$ and~$\rk(\Sigma) = \max\,\{k \in \N \mid
\Sigma_k \neq \emptyset\}$ be the maximal rank of symbols of~$\Sigma$.
Let $X = \{x_i \mid i \in \N\}$ be a countable set of (formal)
variables.  For every~$n\in\N$, we let~$X_n = \{x_i \mid i\in
\nve{n}\}$.  Given a ranked alphabet~$\Sigma$ and a set~$Z$, the
set~$\Tsigma(Z)$ of \emph{$\Sigma$"~trees indexed  by~$Z$} is the
smallest set such that~$Z \subseteq \Tsigma(Z)$ and $\sigma(\seq t1k)
\in \Tsigma(Z)$ for every~$k \in \N$, $\sigma \in\Sigma_k$, and~$\seq
t1k \in \Tsigma(Z)$.  We abbreviate~$\Tsigma(\emptyset)$ simply
by~$\Tsigma$, and any subset~$L\subseteq \Tsigma$ is called a
\emph{tree language}.

Let $\Sigma$~be a ranked alphabet, $Z$~a set, and~$t \in\Tsigma(Z)$.
The set~$\pos(t)$ of \emph{positions of~$t$} is defined
by~$\pos(z) = \{\varepsilon\}$ for all~$z \in Z$ and
by~$\pos(\sigma(\seq t1k)) = \{\varepsilon\} \cup \{iw \mid i \in [k],
w \in \pos(t_i) \}$ for all~$k\in\N$, $\sigma \in \Sigma_k$,
and~$\seq t1k \in \Tsigma(Z)$.  With their help, we define the
\emph{size}~`$\size(t)$' and \emph{height}~`$\he(t)$' of~$t$
as~$\size(t) = \abs{\pos(t)}$
and~$\he(t) = \max_{w \in \pos(t)} \abs w$.  Positions are partially
ordered by the standard prefix order~$\leq$ on~$[\rk(\Sigma)]^*$, and
they are totally ordered by the ascending lexicographic
order~$\preceq$ on~$[\rk(\Sigma)]^*$, in which prefixes are larger;
i.e., $\varepsilon$~is the largest element.  More precisely,
for~$v, w \in \pos(t)$ if there exists~$u \in [\rk(\Sigma)]^*$
with~$vu = w$, then we write~$v\leq w$, call~$v$ a \emph{prefix}
of~$w$, and let~$v^{-1}w = u$ because $u$ is uniquely determined if it
exists.  Provided that~$u = \word u1n$
with~$\seq u1n \in [\rk(\Sigma)]$ we also define the
\emph{path~$[v,\dotsc,w]$ from~$v$ to~$w$} as the
sequence~$(v, vu_1, vu_1u_2, \dotsc, w)$ of positions.  Any two
positions that are~$\leq$"~incomparable are called \emph{parallel}.

Given~$t, t' \in \Tsigma(Z)$ and~$w\in \pos(t)$, the
\emph{label}~$t(w)$ of~$t$ at~$w$, the \emph{subtree}~$t|_w$ of~$t$
at~$w$, and the \emph{substitution}~$t[t']_w$ of~$t'$ into~$t$ at~$w$
are defined by~$z(\varepsilon) = z|_\varepsilon = z$
and~$z[t']_\varepsilon = t'$ for all~$z \in Z$ and by~$t(\varepsilon)
= \sigma$, $t(iw') = t_i(w')$, $t|_\varepsilon = t$, $t|_{iw'} =
t_i|_{w'}$, $t[t']_\varepsilon = t'$, and $t[t']_{iw'} = \sigma
\bigl(\seq t1{i-1}, t_i[t']_{w'}, \seq t{i+1}k \bigr)$ for all
trees~$t = \sigma(\seq t1k)$ with~$k \in \N$, $\sigma \in \Sigma_k$,
$\seq t1k \in\Tsigma(Z)$, all~$i \in [k]$, and all~$w' \in \pos(t_i)$.
For all sets~$S \subseteq \Sigma \cup Z$ of symbols, we let~$\pos_S(t)
= \{w \in \pos(t) \mid t(w) \in S\}$, and we write~$\pos_s(t)$ instead
of~$\pos_{\{s\}}(t)$ for every~$s \in \Sigma \cup Z$.  The set of
variables occuring in~$t$ is~$\var(t) = \{x \in X \mid \pos_x(t) \neq
\emptyset\}$.  Finally, consider~$n \in \N$ and a 
mapping~$\theta' \colon X_n \to \Tsigma(Z)$.  Then by substitution, 
$\theta'$~induces a mapping~$\theta \colon \Tsigma(Z) \to \Tsigma(Z)$
defined by~$\theta(x) = \theta'(x)$ for every~$x \in X_n$, $\theta(z)
= z$ for every~$z \in Z \setminus X_n$, and~$\theta(\sigma(\seq t1k))
= \sigma(\theta(t_1), \dotsc, \theta(t_k))$ for all~$k \in \N$,
$\sigma \in \Sigma_k$, and~$\seq t1k \in \Tsigma(Z)$.  For~$t \in
\Tsigma(Z)$, we denote~$\theta(t)$ by~$t\theta$ or, more commonly, by
$t[x_1 \gets \theta'(x_1), \dotsc, x_n \gets \theta'(x_n)]$. 

Let~$\Box\notin\Sigma$.  A \emph{context} is a tree~$C \in
\Tsigma(\Box)$ with~$\pos_\Box(C) \neq \emptyset$.  More specifically,
we call~$C$ an~\emph{$n$"~context} if~$n=\abs{\pos_\Box(C)}$.  For an
$n$"~context~$C$ and~$\seq t1n \in \Tsigma$, we define the
substitution~$C[\seq t1n]$ as follows.  Let~$\pos_\Box(C) = \{\seq
w1n\}$ be the occurrences of~$\Box$ in~$C$ in lexicographic 
order~$w_1 \prec \dotsb \prec w_n$.  Then we let~$C[\seq t1n] = 
C[t_1]_{w_1} \dotsm [t_n]_{w_n}$.

\subsection*{Tree Homomorphisms and Weighted Tree Grammars}
Given ranked alphabets $\Sigma$~and~$\Gamma$, let~$h' \colon
\Sigma \to T_\Gamma(X)$ be a mapping with $h'(\sigma) \in
T_\Gamma(X_k)$ for all~$k \in \N$ and~$\sigma \in \Sigma_k$.  We
extend~$h'$ to~$h \colon T_\Sigma \to T_\Gamma$ by $h(\alpha) =
h'(\alpha) \in T_\Gamma(X_0) = T_\Gamma$ for all~$\alpha \in \Sigma_0$
and $h(\sigma(\seq t1k)) = h'(\sigma)[x_1 \gets h(t_1), \dotsc, x_k
\gets h(t_k)]$ for all~$k \in \N$, $\sigma \in \Sigma_k$, and~$\seq
t1k \in T_\Sigma$.  The mapping~$h$ is called the \emph{tree
homomorphism induced by~$h'$}, and we identify~$h'$ and its induced 
tree homomorphism~$h$. For the complexity analysis of our decision 
procedure, we define the size of~$h$ as~$\size(h) = \sum_{\sigma \in
  \Sigma} \abs{\pos(h(\sigma))}$. We call~$h$ \emph{nonerasing}  
(respectively, \emph{nondeleting}) if~$h'(\sigma)\notin X$ (respectively,
$\var(h'(\sigma)) = X_k$) for all~$k\in\N$ and~$\sigma \in \Sigma_k$.
In this contribution, we will only consider nonerasing and nondeleting
tree homomorphisms~$h \colon T_\Sigma \to T_\Gamma$, which are
therefore \emph{input finitary}; i.e., the preimage~$h^{-1}(u)$ is finite 
for every~$u \in T_\Gamma$ since~$\abs t \leq \abs u$ for every~$t 
\in h^{-1}(u)$.  Any mapping~$A \colon \Tsigma \to \N$ is called
\emph{$\N$"~weighted tree language}, and we define the weighted tree
language~$h_A \colon T_\Gamma \to \N$ for every~$u \in
T_\Gamma$ by $h_A(u) = \sum_{t \in h^{-1}(u)} A(t)$ and call it the
\emph{image of~$A$ under~$h$}.  This definition relies on the tree
homomorphism to be input-finitary; otherwise the defining sum is  
not finite, so the value~$h_A(u)$ is not necessarily well-defined.

A \emph{weighted tree grammar with equality
  constraints}~(WTGc)~\cite{WTAc-journal} is a 
tuple~$(Q, \Sigma, F, P, \mathord{\wt})$, in which $Q$~is a finite
set of \emph{states}, $\Sigma$~is a ranked alphabet of \emph{input
  symbols}, $F \colon Q \to \N$ assigns a \emph{final weight} to every
state, $P$~is a finite set of \emph{productions} of the form~$(\ell,
q, E)$ with $\ell \in T_\Sigma(Q) \setminus Q$, $q \in Q$, and
finite subset~$E \subseteq \nat^* \times \nat^*$, and
$\mathord{\wt} \colon P \to \N$~assigns a \emph{weight} to every
production.  A production~$p = (\ell, q, E) \in P$ is usually
written~$p = \ell \stackrel E\longrightarrow q$
or~$p = \ell \stackrel E\longrightarrow_{\wt(p)} q$, and the
tree~$\ell$ is called its \emph{left-hand side}, $q$~is its
\emph{target state}, and~$E$ are its \emph{equality constraints},
respectively.  Equality constraints~$(v, v') \in E$ are also
written as~$v = v'$.  A state~$q \in Q$ is \emph{final} if~$F(q) \neq
0$.

Next, we recall the \emph{derivation semantics} of WTGc
from~\cite{WTAc-journal}.  Let~$(v, v') \in \N^* \times \N^*$ be an
equality constraint and~$t \in \Tsigma$.  The tree~$t$
satisfies~$(v,v')$ if and only if~$v, v' \in \pos(t)$
and~$t|_v = t|_{v'}$, and for a finite
set~$C \subseteq \N^* \times \N^*$ of equality constraints, we
write~$t\models C$ if~$t$ satisfies all~$(v,v') \in C$.  Let
$G = (Q, \Sigma, F, P, \mathord{\wt})$ be a WTGc.  A \emph{sentential
  form (for~$G$)} is a tree~$\xi \in T_\Sigma(Q)$.  Given an input
tree~$t \in T_\Sigma$, sentential forms~$\xi, \zeta \in T_\Sigma(Q)$,
a production~$p = \ell \stackrel E\longrightarrow q \in P$, and a
position~$w \in \pos(\xi)$, we
write~$\xi \Rightarrow_{G,t}^{p,w} \zeta$ if~$\xi|_w = \ell$,
$\zeta = \xi[q]_w$, and $t|_w \models E$; i.e., the equality
constraints~$E$ are fulfilled on~$t|_w$.  A
sequence~$d = (p_1, w_1) \dotsm (p_n, w_n) \in (P \times \nat^*)^*$ is
a \emph{derivation (of~$G$) for~$t$} if there
exist~$\seq \xi0n \in T_\Sigma(Q)$ such that~$\xi_0 = t$ and
$\xi_{i-1} \Rightarrow_{G, t}^{p_i,w_i} \xi_i$ for all~$i \in [n]$.
We call $d$ \emph{left-most} if
additionally~$w_1 \prec w_2 \prec \dotsb \prec w_n$.  Note that the
sentential forms~$\seq \xi0n$ are uniquely determined if they exist,
and for any derivation~$d$ for~$t$ there exists a unique permutation
of~$d$ that is a left-most derivation for~$t$.  We call~$d$
\emph{complete} if~$\xi_n \in Q$, and in this case we also call it a
derivation \emph{to}~$\xi_n$.  The set of all complete left-most
derivations for~$t$ to~$q \in Q$ is denoted by~$D^q_G(t)$. A complete
derivation to some final state is called \emph{accepting}. If for every~$p \in
P$, there exists a tree~$t \in T_\Sigma$, a final state~$q$ and a 
derivation~$d = (p_1, w_1) \dotsm (p_m, w_m) \in D_G^q(t)$ such 
that~$F(q) \cdot \wt_{G}(d) \neq 0$ and $p \in \{\seq p1m\}$; i.e.\@
if every production is used in some accepting derivation, then~$G$
is~\emph{trim}.

Let~$d = (p_1, w_1) \dotsm (p_n, w_n)\in D_G^q(t)$ for
some~$t \in \Tsigma$ and~$i \in [n]$.  Moreso, let~$\{\seq j1\ell\}$ 
be the set~$\{j \in [n]\mid w_i \leq w_j\} $ with the indices~$j_1 < 
\dotsm < j_\ell$ of those positions of which $w_i$~is a prefix.  We refer
to~$(p_{j_1},w_i^{-1} w_{j_1} ), \dotsc,(p_{j_\ell},
w_i^{-1} w_{j_\ell} ) $ as the~\emph{derivation for~$t|_{w_i}$
incorporated in~$d$}. Conversely, for~$w \in \N^*$ we abbreviate the
derivation~$(p_1,ww_1) \dotsm (p_n,ww_n)$ by~$wd$.

The \emph{weight} of a derivation~$d = (p_1,w_1) \dotsm (p_n, w_n)$
is defined as~$\wt_G(d) = \prod_{i = 1}^n \wt(p_i)$.  The
weighted tree language generated by~$G$, written~$\sem G \colon
\Tsigma \to \N$, is defined for all~$t \in \Tsigma$~by
\[ \sem G(t) = \sum_{q \in Q,\, d \in D^q_G(t)} F(q) \cdot \wt_G(d)
\enspace. \] 
For~$t \in \Tsigma$ and~$q \in Q$, we will often use the 
value~$\wt_G^q(t)$ defined as~$\wt_G^q(t) = \sum_{d\in D_G^q(t)}
\wt_G(d)$.  Using distributivity, $\sem G(t)$ then simplifies
to~$\sem G(t) = \sum_{q\in Q} F(q) \cdot \wt_G^q(t)$. 
We call two WTGc \emph{equivalent} if they generate the same weighted
tree language. 

We call a WTGc~$(Q, \Sigma, F, P, \mathord{\wt})$ a \emph{weighted
  tree grammar}~(WTG) if~$E = \emptyset$ for every
production~$\ell \stackrel E\longrightarrow q \in P$; i.e., no
production utilizes equality constraints.  Instead of~$\ell
\stackrel \emptyset\longrightarrow q$ we also simply
write~$\ell\to q$.  Moreover, we call a WTGc a \emph{weighted tree
  automaton with equality constraints} (WTAc) if~$\pos_\Sigma(\ell) =
\{\varepsilon\}$ for every production~$\ell \stackrel E\longrightarrow
q \in P$, and a \emph{weighted tree automaton} (WTA) if it is both a
WTG and a WTAc. 
The classes of WTGc and WTAc are equally 
expressive, and they are strictly more expressive than the class of WTA 
\cite{WTAc-journal}. We call a weighted tree language \emph{regular} if 
it is generated by a WTA and \emph{constraint-regular} if it is generated
by a WTGc.  Productions with weight~$0$ are obviously useless, 
so we may assume that~$\wt(p) \neq 0$ for every production~$p$. Finally, 
we define the size of a WTGc as follows.

\begin{definition}
  Let~$G = (Q, \Sigma, F, P, \mathord{\wt})$ be a WTGc and~$p = \ell 
  \stackrel E\longrightarrow q\in P$ be a production.  We define 
  the~\emph{height of~$p$} as~$\he(p) = \he(\ell)$ and 
  its~\emph{size} as~$\size(p) = \abs{\pos(\ell)}$, the~\emph{height
    of~$P$} as~$\he(P) = \max_{p\in P} \he(p)$ and its~\emph{size}
  as~$\size(P) = \sum_{p\in P} \size(p)$, and finally the~\emph{height
    of~$G$} as~$\he(G) = \abs Q \cdot \he(P)$ and its~\emph{size}
  as~$\size(G) = \abs Q + \size(P)$.  
\end{definition}

It is known~\cite{WTAc-journal} that WTGc can be used to represent
homomorphic images of regular weighted tree languages.  Let $A \colon
T_\Sigma \to \N$ be a regular weighted tree language
(effectively given by a WTA) and~$h \colon T_\Sigma \to T_\Gamma$ be a
tree homomorphism.  Following~\cite[Theorem~5]{WTAc-journal} we 
can construct a WTGc~$G = (Q, \Gamma, F,P, \mathord{\wt})$ of a
specific shape such that~$\sem G = h_A$.  More precisely, the constructed 
WTGc~$G$ has a designated nonfinal~\emph{sink state}~$\bot \in Q$ such 
that~$F(\bot) = 0$ as well as $p_\gamma = \gamma(\bot, \dotsc,
\bot) \to \bot \in P$ and~$\wt(p_\gamma) = 1$ for every~$\gamma \in
\Gamma$.  In addition, every production~$p = \ell
\stackrel E\longrightarrow q \in P$ satisfies the following two
properties.  First,~$E \subseteq \pos_Q(\ell)^2$;
i.e., all equality constraints point to the $Q$"~labeled positions of
its left-hand side.  Without loss of generality, we can assume that the
set~$E$ of equality constraints is reflexive, symmetric, and
transitive; i.e., an equivalence relation on a subset~$D \subseteq
\pos_Q(\ell)$, so not all occurrences of states need to be
constrained.  Second, $\ell(v) = \bot$ and~$\ell(w) \neq \bot$ for
every~$v \in [w']_E \setminus \{w\}$ and~$w' \in D$, where~$w =
\min_{\preceq} [w']_E$; i.e., all but the lexicographically least
position in each equivalence class of~$E$ are guarded by state~$\bot$.
Essentially, the WTGc~$G$ performs its checks (and charges weights)
exclusively on the lexicographically least occurrences of
equality-constrained subtrees.  All the other subtrees, which by means
of the constraint are forced to coincide with another subtree, are
simply ignored by the WTGc, which formally means that they are
processed in the designated sink state~$\bot$. In the following, we will 
use~$\bot$ to indicate such a sink state, and 
write~$Q \cup \{\bot\}$ to explicitly indicate its presence.  

In~\cite{WTAc-journal} WTGc of the special shape just discussed were 
called eq"~restricted, but since these will be the primary objects of interest 
in this work, we simply call them WTGh here. The constructive 
proof of the following statement can be found in the appendix.

\begin{restatable}[see~\protect{\cite[Theorem~5]{WTAc-journal}}]{theorem}{thmhom}
  \label{thm:hom}
  Let $G = (Q, \Sigma, F, P, \mathord{\wt})$ be a trim WTA and $h
  \colon T_\Sigma \to T_\Gamma$ be a nondeleting and nonerasing tree
  homomorphism.  Then there exists a trim WTGh~$G'$
  with~$\sem{G'}=h_{\sem G}$.  Moreover, $\size(G') \in \mathcal{O}
  \bigl(\size(G) \cdot \size(h) \bigr)$ and $\he(G') \in \mathcal{O}
  \bigl(\size(h) \bigr)$.
\end{restatable}

\begin{example}
  \label{ex: WTGh}
  Let~$G=(Q \cup \{\bot\}, \Gamma, F, P, \mathord{\wt})$
  with~$Q = \{q, q_f\}$, $\Gamma = \{\alpha^{(0)}, \gamma^{(1)}, \delta^{(3)}\}$, $F(q) = F(\bot) = 0$ and~$F(q_f) = 1$, and the
  following set~$P$ of productions.
  \[ \Bigl\{\alpha \to_1 q,\; \gamma(q) \to_2 q,\; \delta \bigl(q,
    \gamma(\bot), q \bigr) \stackrel{1=21}\longrightarrow_1 q_f,\quad
    \alpha \to_1 \bot,\; \gamma(\bot) \to_1 \bot,\; \delta(\bot, \bot,
    \bot) \to_1 \bot \Bigr\} \] The WTGc~$G$ is a WTGh.  It generates
  the homomorphic image~$\sem G = h_A$ for the tree homomorphism~$h$
  induced by the mapping~$\alpha \mapsto \alpha$,
  $\gamma \mapsto \gamma(x_1)$, and
  $\sigma \mapsto \delta \bigl(x_2, \gamma(x_2), x_1 \bigr)$ applied
  to the regular weighted tree language~$A \colon T_\Sigma \to \N$
  given by~$A(t) = 2^{\abs{\pos_\gamma(t)}}$ for
  every~$t \in T_\Sigma$
  with~$\Sigma = \{\alpha^{(0)}, \gamma^{(1)}, \sigma^{(2)}\}$. The
  weighted tree language~$\sem G$ is itself not regular because its
  support is clearly not a regular tree language.
\end{example}

The restrictions in the definition of a WTGh allow us to trim it effectively 
using a simple reachability algorithm. For more details, we refer the 
reader to the appendix.

\begin{restatable}{lemma}{lmtrim}
  \label{lm:trim}
  Let~$G = (Q \cup \{\bot\}, \Sigma, F, P, \mathord{\wt})$ be a WTGh.
  An equivalent, trim WTGh~$G'$ can be constructed in polynomial
  time.
\end{restatable}

\section{Substitutions in the Presence of Equality Constraints}
This short section recalls from~\cite{WTAc-journal} some definitions
together with a pumping lemma for WTGh, which will be essential for
deciding the integer-weighted HOM-problem. First, we need to refine 
the substitution of trees such that it complies with existing constraints.

\begin{definition}[\protect{see~\cite{WTAc-journal} and
  cf.~\cite{godoy2013hom}}]
  \label{def: pumping} 
  Let~$G= (Q \cup \{\bot\}, \Sigma, F, P, \mathord{\wt})$ be a WTGh,
  and let~$d = (p_1, w_1) \dotsm (p_m, w_m)$ be a complete left-most
  derivation for a tree~$t \in \Tsigma$ to a state~$q \in Q
  \cup \{\bot\}$.  Furthermore, let~$j \in \nve{m}$ such that~$q_j
  = \bot$ if and only if~$q = \bot$; i.e., the target state~$q_j$ of
  production~$p_j$ is~$\bot$ if and only if~$q = \bot$.  We note that
  automatically $q_j = \bot$ whenever~$q = \bot$.  Finally, let
  $d'$~be a derivation to~$q_j$ for some tree~$t' \in \Tsigma$,
  and let~$d'_{\bot}$ be the derivation of~$G$ for~$t'$ where every
  occurring state is~$\bot$.  We define the substitution~$d
  \sem{d'}_{w_j}$ of~$d'$ into~$d$ at~$w_j$ recursively as follows.
  \begin{itemize}
  \item If~$w_j = \varepsilon$ (i.e., $j = m$), then we define $d
    \sem{d'}_{w_j} = d'$.
  \item Otherwise, let $p_m = \ell \stackrel E\longrightarrow q$ be
    the production utilized last, $\pos_Q(\ell) = \{\seq v1n\}$,
    and let~$\seq d1n$ be the derivations for~$\seq{t|}{v_1}{v_n}$
    incorporated in~$d$, respectively.  Obviously there exists~$s \in
    \nve{n}$ such that~$v_s \leq w_j$.  Let~$\hat{w} = v_s^{-1}w_j$,
    which is a position occurring in~$d_s$.  Correspondingly, we
    define~$d'_s = d_s \sem{d'}_{\hat{w}}$ and for every~$i \in\nve{n}
    \setminus \{s\}$, we define~$d'_i = d_i \sem{d'_{\bot}}_{\hat{w}}$
    if~$(v_i, v_s) \in E$ and otherwise~$d'_i = d_i$.  Then~$d
    \sem{d'}_{w_j}$ is obtained by reordering the derivation~$(v_1
    d'_1) \dotsm (v_n d'_n)(p_m,w_m)$ such that it is left-most.
  \end{itemize}
  The tree derived by~$d \sem{d'}_{w_j}$ is denoted by~$t \sem
  {t'}_{w_j}^d$ or simply~$t\sem{t'}_{w_j}$, if the original derivation 
  for~$t$ is clear from the context.
\end{definition}

\begin{example}
  \label{ex: WGTh and subst}
  Consider the WTGh~$G$ of~Example~\ref{ex: WTGh} and the 
  following tree~$t$ it generates into which we want to substitute the
  tree~$t' = \gamma(\alpha)$ at position~$w = 11$.
  
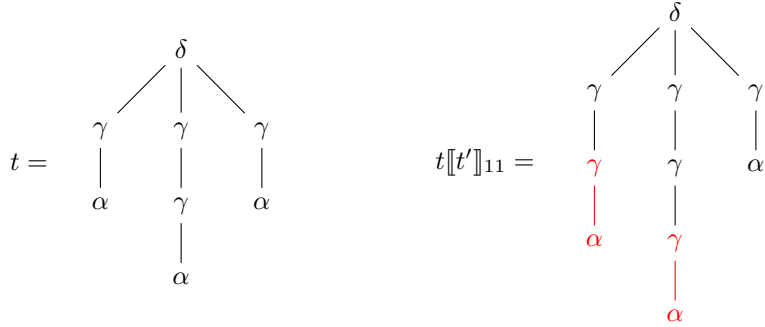
\begin{center}
    \begin{tikzpicture}
      \node at (-2,0) (label) {$t =$};
      \node at (0, 0)  (a) {
        \begin{forest}
          for tree={%
            l sep=0.1cm,
            s sep=0.6cm,
            minimum height=0.000008cm,
            minimum width=0.000015cm,
          }
          [$\delta$
          [$\gamma$[$\alpha$]]
          [$\gamma$[$\gamma$[$\alpha$]]]
          [$\gamma$[$\alpha$]]			
          ]
        \end{forest}
      };
      \node at (4,0) (label) {$t \sem{t'}_{11} =$};
      \node at (6.5,0)  (a) {
        \begin{forest}
          for tree={%
            l sep=0.1cm,
            s sep=0.6cm,
            minimum height=0.000008cm,
            minimum width=0.000015cm,
          }
          [$\delta$
          [$\gamma$[$\textcolor{red}{\gamma}$[$\textcolor{red}{\alpha}$,
          edge={red}]]]
          [$\gamma$[$\gamma$[$\textcolor{red}{\gamma}$[$\textcolor{red}{\alpha}$, edge={red}]]]]
          [$\gamma$[$\alpha$]]			
          ]
        \end{forest}
      };
    \end{tikzpicture} 
  \end{center}
  We consider the following complete left-most derivation for~$t$
  to~$q_f$.
  \begin{linenomath*}
    \begin{align*}
      d
      &= \Bigl(\alpha \to q, 11\Bigr)\, \Bigl(\gamma(q) \to q, 1\Bigr)
      \phantom{{}={}} \Bigl(\alpha \to \bot, 211 \Bigr)\,
        \Bigl(\gamma(\bot) \to \bot, 21\Bigr) \\*
      &\phantom{{}={}} \Bigl(\alpha \to q, 31 \Bigr) \, \Bigl(\gamma(q)
        \to q, 3\Bigr) \, \Bigl(\delta\bigl(q, \gamma(\bot), q \bigr)
        \stackrel{1=21}\longrightarrow q_f, \varepsilon\Bigr)
    \end{align*}
  \end{linenomath*}
  Moreover, let~$d' =
  \bigl(\alpha \to q, 1\bigr)\, \bigl(\gamma(q) \to q,
  \varepsilon\bigr)$ and~$d'_\bot = \bigl(\alpha \to \bot, 1\bigr)\,
  \bigl(\gamma(\bot) \to \bot, \varepsilon \bigr)$.  With the
  notation of Definition~\ref{def: pumping}, in the first step we
  have~$v_1 = 1$, $v_2 = 21$, $v_3 = 3$, $d_1 = d_3 = d'$, $d_2 =
  d'_{\bot}$, and $\hat{w} = v_1^{-1} w = 1$.  Respecting the only
  constraint~$1 = 21$, we set~$d'_1 = d_1 \sem{d'}_{\hat{w}} =
  d' \sem{d'}_1$, $d'_2 = d_2 \sem{d'_{\bot}}_{\hat{w}} = d'_{\bot}
  \sem{d'_{\bot}}_1$, and~$d_3' = d_3 = d'$.  Eventually, $d'_1\! =\! (\alpha 
  \to q, 11) (\gamma(q) \to q, 1)  
  (\gamma(q) \to q, \varepsilon)$ and $d'_2 \!=\! (\alpha \to \bot,
  11)(\gamma(\bot) \to \bot, 1)  (\gamma(\bot) \to \bot,
  \varepsilon)$.  Hence, we obtain the following derivation~$d
  \sem{d'}_{11}$ for our new tree~$t \sem{t'}_{11}$.
  \begin{linenomath*}
    \begin{align*}
      d \sem{d'}_{11}
      &= \Bigl(\alpha \to q, 111\Bigr) \, \Bigl(\gamma(q) \to q, 11
        \Bigr) \, \Bigl(\gamma(q) \to q, 1\Bigr)
        \, \Bigl(\alpha \to \bot, 2111 \Bigr) \,
        \Bigl(\gamma(\bot) \to \bot, 211\Bigr) \\*
         &\phantom{{}={}} \Bigl(\gamma(\bot) \to
        \bot, 21 \Bigr) 
      \phantom{{}={}} \Bigl(\alpha \to q, 31 \Bigr) \, \Bigl(\gamma(q)
        \to q, 3 \Bigr) \, \Bigl(\delta\bigl(q, \gamma(\bot), q\bigr)
        \stackrel{1=21}\longrightarrow q_f, \varepsilon \Bigr)
    \end{align*}
  \end{linenomath*}
  Although~$t|_{31} = \alpha$ also coincides with the
  subtree~$t|_{11} = \alpha$ we 
  replaced, these two subtrees are not equality-constrained, so the
  simultaneous substitution does not affect~$t|_{31}$.
\end{example}

The substitution of Definition~\ref{def: pumping} allows us to prove a
pumping lemma for the class of WTGh: 
%
%
If~$d$~is an accepting derivation of a
WTGh~$G = (Q \cup \{\bot\}, \Sigma, F, P, \mathord{\wt})$ for a
tree~$t$ with~$\he(t) > \he(G)$, then there exist at
least~$\abs{Q\setminus\{\bot\}} +1$ 
positions~$w_1 > \dotsb > w_{\abs Q+1}$ in~$t$ at which $d$~applies
productions with non-sink target states.  By the pigeonhole principle,
there thus exist two positions~$w_i > w_j$ in~$t$ at which $d$~applies
productions with the same non-sink target state.  Employing the
substitution we just defined, we can substitute~$t|_{w_j}$ into~$w_i$
and obtain a derivation of~$G$ for~$t \sem{t|_{w_j}}_{w_i}$.  This
process can be repeated to obtain an infinite sequence of trees
strictly increasing in size.  Formally, the following lemma was proved 
in~\cite{WTAc-journal}.

\begin{lemma}[\protect{\cite[Lemma~4]{WTAc-journal}}]
  \label{lm:existence of subs}
  Let~$G = (Q \cup \{\bot\}, \Sigma, F, P, \mathord{\wt})$ be a WTGh.
  Consider some tree~$t \in T_\Sigma$ and non-sink state~$q \in Q
  \setminus \{\bot\}$ such that~$\he(t) > \he(G)$ and~$D_G^q(t) \neq
  \emptyset$.  Then there are infinitely many pairwise distinct
  trees~$t_0, t_1, \dotsc$ such that~$D_G^q(t_i) \neq \emptyset$ for
  all~$i \in \N$.
\end{lemma}

\begin{example}
  Recall the WTGh~$G$ of Example~\ref{ex: WTGh}.  We have~$\he(P) = 2$
  and~$\he(G) = 4$, but for simplicity, we choose the smaller tree~$t
  = \delta(\gamma(\alpha), \gamma(\gamma(\alpha)), \gamma(\alpha))$,
  which we also considered in Example~\ref{ex: WGTh and subst},
  since it also allows pumping.  The derivation~$d$ presented in
  Example~\ref{ex: WGTh and subst} for~$t$ applies the
  productions~$(\alpha \to q)$ at~$11$ and $\gamma(q) \to q$ at~$1$,
  so we substitute~$t|_1 = \gamma(\alpha)$ at~$11$ to obtain~$t
  \sem{\gamma(\alpha)}_{11}$.  In fact, this is exactly the
  substitution we illustrated in Example~\ref{ex: WGTh and subst}.
\end{example}

\section{The Decision Procedure}
Let us now turn to the $\N$"~weighted version of the HOM-problem.
In the following, we show that the regularity of the homomorphic image
of a regular $\N$-weighted tree language is decidable in polynomial 
time. More precisely, we prove the following theorem.

\begin{theorem}
  \label{thm: main}
  The weighted HOM-problem over~$\N$ is polynomial; i.e.\@ for fixed
  ranked alphabets~$\Gamma$~and~$\Sigma$, given a trim WTA~$H$
  over~$\Gamma$, and a nondeleting, nonerasing tree homomorphism~$h
  \colon T_\Gamma \to \Tsigma$, it is decidable in polynomial time
  whether 
  $h_{\sem H}$~is regular.
\end{theorem}

For the proof, we follow the general outline of the unweighted
case~\cite{godoy2013hom}.  Given a regular weighted tree language~$A$ 
(represented by a trim WTA) and a tree homomorphism~$h$, we begin by 
constructing a trim WTGh~$G$ for its image~$\sem G = h_A$ 
applying~\cref{thm:hom}.  We then show that $\sem G$~is regular if and
only if for all derivations of~$G$ the equality constraints occurring
in the derivation only apply to subtrees of height at most~$\he(G)$.
In other words, if there exists a
production~$\ell\stackrel E\longrightarrow q$ in~$G$ such that for
some equality constraint~$(u,v) \in E$ with non-sink
state~$q = \ell(u)$ there exists a tree~$t \in T_\Sigma$
with~$\he(t) > \he(G)$ and~$D_G^q(t) \neq \emptyset$, then $\sem G$~is
not regular, and if no such production exists, then $\sem G$~is
regular.  There are thus three parts to our proof.  First, we show
that the existence of such a production is decidable in polynomial
time. Then we show that $\sem G$~is regular if no such production
exists.  Finally, we show that $\sem G$~is not regular if such a
production exists.  For convenience, we attach a name to the property
described here.

\newcommand{\ldpp}{large duplication property}
\begin{definition}
  Let~$G = (Q \cup \{\bot\}, \Sigma, F, P, \mathord{\wt})$ be a trim
  WTGh.  We say that $G$~has the \emph{\ldpp{}} if there exist a
  production~$\ell\stackrel E\longrightarrow q \in P$, an equality
  constraint~$(u,v) \in E$ with~$\ell(u) \neq \bot = \ell(v)$, and a
  tree~$t \in T_\Sigma$ such that~$\he(t) > \he(G)$
  and~$D_G^{\ell(u)}(t) \neq \emptyset$.
\end{definition}

We start with the decidability of the \ldpp{}.

\begin{lemma}
  \label{lm: ldpp decid}
  Consider a fixed ranked alphabet~$\Sigma$.  The following is
  decidable in polynomial time:  Given a trim WTGh~$G$,
  does it satisfy the \ldpp{}?
\end{lemma}

\begin{proof}
  Let $G = (Q \cup \{\bot\}, \Sigma, F, P, \mathord{\wt})$ and
  construct the directed graph~$G = (Q, E)$ with
  edges~$E = \bigcup_{\ell \stackrel E\longrightarrow q \in P} \{ (q',
  q) \mid q' \in Q, \pos_{q'}(\ell) \neq \emptyset\}$.  Clearly, the
  \ldpp{} is equivalent to the condition that there exists a
  production~$\ell \stackrel E\longrightarrow q \in P$, an equality
  constraint~$(u, v) \in E$ with~$\ell(u) \neq \bot = \ell(v)$, and a
  state~$q' \in Q \setminus \{\bot\}$ such that there exists a cycle
  from~$q'$ to~$q'$ in~$G$ and a path from~$q'$ to~$q$ in~$G$.  This
  equivalent condition can be checked in polynomial time.  The
  equivalence of the two statements is easy to establish.  If the
  \ldpp{} holds, then the pumping lemma~\cite[Lemma~4]{WTAc-journal}
  exhibits the required cycle and path.  Conversely, if the cycle and
  path exist, then the pumping lemma~\cite[Lemma~4]{WTAc-journal} can
  be used to derive arbitrarily tall trees for which a derivation
  exists.
\end{proof}

Next, we show that a WTGh~$G$ generates regular~$\sem G$ if it does
not satisfy the \ldpp{}.  To this end, we construct the
\emph{linearization} of~$G$.  The linearization of a WTGh~$G$ is a WTG
that simulates all derivations of~$G$ which only ensure the
equivalence of subtrees of height at most~$\he(G)$.  This is achieved
by replacing every production~$\ell \stackrel E\longrightarrow q$
in~$G$ by the collection of all productions~$\ell' \to q$ which can be
obtained by substituting each position constrained by~$E$ with a
compatible tree of height at most~$\he(G)$ that satisfies the equality
constraints of~$E$.  Note that positions in~$\ell$ that are unconstrained 
by~$E$ are unaffected by these substitutions.  Formally, we define the 
linearization following~\cite[Definition~7.1]{godoy2013hom}.

\newcommand{\ling}{\mathrm{lin}(G)}
\begin{definition}
  Let~$G = (Q \cup \{\bot\}, \Sigma, F, P, \mathord{\wt})$ be a
  WTGh.  The linearization~$\ling$ of~$G$ is the WTG~$\ling = (Q \cup
  \{\bot\}, \Sigma, F, P_{\mathrm{lin}},
  \mathord{\wt_{\mathrm{lin}}})$, where
  $P_{\mathrm{lin}}$~and~$\wt_{\mathrm{lin}}$ are defined as follows.
  For~$\ell' \in T_\Sigma(Q) \setminus Q$ and~$q \in Q$, we
  let~$(\ell' \to q) \in P_{\mathrm{lin}}$ if and only if there exist
  a production~$(\ell \stackrel E\longrightarrow q) \in P$,
  positions~$\seq w1k \in \pos_{Q \cup \{\bot\}}(\ell)$, and
  trees~$\seq t1k \in T_\Sigma$ with
  \begin{itemize}
  \item $\{\seq w1k\} = \bigcup_{w \in \pos_\bot(\ell)} [w]_E$; i.e.,
    $E$~constrains exactly the positions~$\seq w1k$,
  \item $t_i = t_j$ if~$(w_i, w_j) \in E$ for all~$i,j \in \nve{k}$,
  \item $\ell' = \ell[t_1]_{w_1} \dotsm [t_k]_{w_k}$, and
  \item $D_G^{\ell(w_i)}(t_i) \neq \emptyset$ and~$\he(t_i) \leq
    \he(G)$ for all~$i \in \nve{k}$.
  \end{itemize}
  For every such production~$\ell' \to q$ we
  define~$\wt_{\mathrm{lin}}(\ell' \to q)$ as the sum over all
  weights
  \[ \wt(\ell \stackrel E\longrightarrow q) \cdot \prod_{i \in
      \nve{k}} \wt_G^{\ell({w_i})}(t_i) \] for all~$(\ell \stackrel
  E\longrightarrow q) \in P$, $\seq w1k \in \pos_{Q \cup
    \{\bot\}}(\ell)$, and~$\seq t1k \in T_\Sigma$ as above.
\end{definition}

If a trim WTGh~$G$ does not satisfy the \ldpp{}, then every equality
constraint in every derivation of~$G$ only ensures the equality of
subtrees of height at most~$\he(G)$.  Thus, $\ling$~and~$G$ generate
the same weighted tree language~$\sem G = \sem{\ling}$, which is then
regular because $\ling$~is a WTG. Thus we summarize:
\begin{proposition}~\label{prop: no ldpp reg}
Let~$G$ be a trim WTGh and suppose that~$G$ does not satisfy the 
\ldpp{}. Then~$\sem G$ is a regular weighted tree language.
\end{proposition}

Finally, we show that if a WTGh~$G = (Q \cup \{\bot\}, \Sigma, F, P,
\mathord{\wt})$ satisfies the \ldpp{}, then $\sem G$~is not regular.
For this, we first show that if $G$~satisfies the
\ldpp{},  then we can decompose it into two WTGh $G_1$~and~$G_2$ such
that~$\sem G = \sem{G_1} + \sem{G_2}$ and at least one of
$\sem{G_1}$~and~$\sem{G_2}$ is not regular.  To conclude the desired
statement, we then show that the sum~$\sem G = \sem{G_1} + \sem{G_2}$
is also not regular.  For the decomposition, consider the following
idea.  Assume that there exists a production~$p = (\ell \stackrel
E\longrightarrow q) \in P$ as in the \ldpp{} such that~$F(q) \neq 0$.
Then we create two copies $G_1$~and~$G_2$ of~$G$ as follows.  In~$G_1$
we set all final weights to~$0$, add a new state~$f$ with final
weight~$F(q)$, and add the new production~$(\ell \stackrel
E\longrightarrow f)$ with the same weight as~$p$.  On the other hand,
in~$G_2$ we set the final weight of~$q$ to~$0$, add a new state~$f$
with final weight~$F(q)$, and for every production~$p' = (\ell'
\stackrel{E'}\longrightarrow q) \in P$ except~$p$, we add the new
production~$\ell' \stackrel{E'}\longrightarrow f$ to~$G_2$ with the
same weight as~$p'$.  Then~$\sem G = \sem{G_1} + \sem{G_2}$ 
because every derivation of~$G$ whose last production is~$p$ is now a
derivation of~$G_1$ to~$f$, and every other derivation is either
directly a derivation of~$G_2$ or, in case of other derivations
to~$q$, is a derivation of~$G_2$ to~$f$.

By our assumption on the production~$p = (\ell \stackrel
E\longrightarrow q)$, there exist a tall tree~$t \in T_\Sigma$
with~$\he(t) > \he(G)$ and a constraint~$(u,v) \in E$ with $\ell(u)
\neq \bot = \ell(v)$ and~$D_G^{\ell(u)}(t) \neq \emptyset$.  Thus,
every tree~$t'$ generated by~$G_1$ satisfies~$t'|_u = t'|_v$, and by
Lemma~\ref{lm:existence of subs}, there exist infinitely many pairwise
distinct trees with a derivation to~$\ell(u)$.  The support (i.e., set
of nonzero weighted trees) of $\sem{G_1}$ is therefore not a regular
tree language.  This implies that $\sem{G_1}$~is not regular as the
support of every regular weighted tree language over~$\N$ is a regular
tree language~\cite{fulvog09}. 

In general, we cannot expect that a production~$\ell \stackrel
E\longrightarrow q$ as in the \ldpp{} exists with~$F(q) \neq 0$. For details
on the general case, we refer the reader to the appendix. 

\begin{restatable}{lemma}{lmdecomp}
  \label{lm: decomp}
  Let $G = (Q \cup \{\bot\}, \Sigma, F, P, \mathord{\wt})$ be a trim
  WTGh that satisfies the \ldpp{}.  Then there exist two trim WTGh 
  $G_1 = (Q_1 \cup \{\bot\}, \Sigma, F_1, P_1, \mathord{\wt_1})$ and
  $G_2 = (Q_2 \cup \{\bot\}, \Sigma, F_2, P_2, \mathord{\wt_2})$
  such 
  that~$\sem G=\sem{G_1}+\sem{G_2}$ and for some~$f\in Q_1$ we have 
  \begin{itemize}
  \item $F_1(f) \neq 0$ and~$F_1(q) = 0$ for
    all~$q \in Q_1 \setminus \{f\}$, and 
  \item there exists exactly one production~$p_{\text{\upshape f}} =
    (\ell_{\text{\upshape f}} \stackrel{E_{\text{\upshape f}}}\longrightarrow
    f) \in P_1$ with target state~$f$, and for this production there
    exists~$(u, v) \in E_{\text{\upshape f}}$ with~$\ell_{\text{\upshape f}}(u) 
    \neq\ell_{\text{\upshape f}}(v) = \bot$ and an infinite sequence of 
    pairwise distinct trees~$t_0, t_1, t_2, \dotsc \in T_\Sigma$ such
    that~$D_{G_1}^{\ell_{\text{\upshape f}}(u)}(t_i) \neq \emptyset$ for all~$i 
    \in \N$.
  \end{itemize}
\end{restatable}

\begin{example}
  We present an example for the decomposition in Lemma~\ref{lm:
    decomp}.  Consider the trim
  WTGh~$G = (Q \cup \{\bot\}, \Sigma, P, F, \mathord{\wt})$ with~$Q =
  \{q_0, \bar{q}, q_{\text{\upshape f}} \}$, $\Sigma=\{\alpha^{(0)}, \gamma^{(1)},
  \sigma^{(2)}, \gamma_1^{(1)}, \gamma_2^{(1)}\}$, final
  weights~$F(q_{\text{\upshape f}} ) = 1$ and~$F(q_0) = F(\bar{q}) = F(\bot) = 0$,
  and the set~$P=P_\bot\cup P'$ defined by
   $   P'=\bigl\{ \,
       \alpha
      \to_1 q_0, \;
      \gamma(q_0)
      \to_1 q_0,\;
      \sigma(q_0, \bot)
      \stackrel{1=2}{\longrightarrow}_2 \bar{q},\;
      \gamma_1(\bar{q})
      \to_2 \bar{q},\;
      \gamma_2(\bar{q})
      \to_2 \bar{q},\;
      \sigma(\bar{q}, q_0)
      \to_2 q_{\text{\upshape f}} \,\bigr\}$
 and the usual productions targeting~$\bot$
  in~$P_\bot$.  Trees of the
  form~$\gamma(\dotsm(\gamma(\alpha))\dotsm)$ of arbitrary height are
  subject to the constraint~$1=2$, so $G$~satisfies the \ldpp{}.
    
  We consider~$t'$ as in Figure~\ref{fig: ex
    decomp} and use its (unique) derivation in~$G$. 
    Following the approach sketched above,
  we choose a new state~$f$ and define~$G_1 = (Q \cup \{f\} \cup
  \{\bot\}, \Sigma, F_1, P_1, \mathord{\wt_1})$, where~$F_1(f) = 1$
  and~$F_1(q) = 0$ for every~$q \in Q \cup \{\bot\}$,
  and~$P_1 = P\cup \{p_{\text f}\}$ with the new production~$p_{\text
    f}$ depicted in Figure~\ref{fig: ex decomp}, which joins all the
  productions of~$G$ used to derive~$t'$, from the one evoking the
  \ldpp{} to the one targeting a final state.
  \begin{figure}
    \begin{center}
      \caption{The tree~$t'$ and the new production~$p_f$}
      \label{fig: ex decomp}
      \begin{tikzpicture}
        \node at (-1.5,0) (label) {$t' =$};
        \node at (0, -1.5)  (a) {
          \begin{forest}
            for tree={%
              l sep=0.1cm,
              s sep=0.6cm,
              minimum height=0.000008cm,
              minimum width=0.000015cm,
            }
            [$\sigma$[$\gamma_1$[$\sigma$[$\alpha$
            ][$\alpha$]
            ]
            ][$\alpha$]	
            ]
          \end{forest}
        };
        \node at (5,0) (label) {$p_{\text f} = $};
        \node at (6.5, -1.5)  (a) {
          \begin{forest}
            for tree={%
              l sep=0.1cm,
              s sep=0.6cm,
              minimum height=0.000008cm,
              minimum width=0.000015cm,
            }
            [$\sigma$[$\gamma_1$[$\sigma$[$q_0$
            ][$\bot$]
            ]
            ][$q_0$]	
            ]
          \end{forest}
        };
        \node at (8.5,0) (label){$\stackrel{111=112}{\longrightarrow}_8 
          \quad f$};
      \end{tikzpicture}
    \end{center}
  \end{figure}
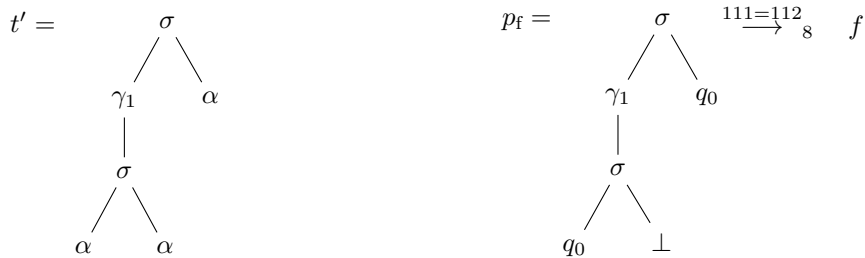
  It remains to construct a WTGh~$G_2$ such that~$\sem G= \sem{G_1} +
  \sem{G_2}$. All productions of~$G$ still occur
  in~$G_2$, but~$q_{\text{\upshape f}} $ is not final anymore.  Instead, we
  add a state~$f$ with~$F_2(f) = F(q_{\text{\upshape f}} ) = 1$ and make sure that
  this state adopts all other accepting derivations that formerly led
  to~$q_{\text{\upshape f}} $. 
  For this, we handle first the derivations
  that coincide with the derivation for~$t'$ at the juncture
  positions~$\varepsilon$ and~$1$, but not at~$2$. This leads to the
  following new productions~$p_1^1$ and~$p_2^1$:
  \begin{center}
    \begin{tikzpicture}
      \node at (-1.5,0) (label) {$p^1_1 =$};
      \node at (0, -1.5)  (a) {
        \begin{forest}
          for tree={%
            l sep=0.1cm,
            s sep=0.6cm,
            minimum height=0.000008cm,
            minimum width=0.000015cm,
          }
          [$\sigma$[$\gamma_1$[$\gamma_1$[$\bar{q}$
          ]
          ]
          ][$q_0$]	
          ]
        \end{forest}
      };
      \node at (1.5,0) (label){$\to_8 \; f$};
      \node at (5.5,0) (label) {$p^1_2 =$};
      \node at (7, -1.5)  (a) {
        \begin{forest}
          for tree={%
            l sep=0.1cm,
            s sep=0.6cm,
            minimum height=0.000008cm,
            minimum width=0.000015cm,
          }
          [$\sigma$[$\gamma_1$[$\gamma_2$[$\bar{q}$
          ]
          ]
          ][$q_0$]	
          ]
        \end{forest}
      };
      \node at (8.5,0) (label){$\to_8 \; f\enspace.$};
    \end{tikzpicture}
  \end{center}
  Next  we cover the derivations that differ from the
  derivation for~$t'$ at the position~$1$ but coincide with it at the
  root.  This leads to the new productions
  \begin{center}
    \begin{tikzpicture}
      \node at (-1.5,0) (label) {$p^2_1 =$};
      \node at (0, -1)  (a) {
        \begin{forest}
          for tree={%
            l sep=0.1cm,
            s sep=0.6cm,
            minimum height=0.000008cm,
            minimum width=0.000015cm,
          }
          [$\sigma$[$\sigma$[$q_0$][$\bot$]
          ][$q_0$]
          ]
        \end{forest}
      };
      \node at (2,0) (label){$\stackrel{11=12}{\longrightarrow}_4 
        \quad f$};
      \node at (5.5,0) (label) {$p^2_2 =$};
      \node at (7, -1)  (a) {
        \begin{forest}
          for tree={%
            l sep=0.1cm,
            s sep=0.6cm,
            minimum height=0.000008cm,
            minimum width=0.000015cm,
          }
          [$\sigma$[$\gamma_2$[$q_0$
          ]
          ][$q_0$]	
          ]
        \end{forest}
      };
      \node at (8.5,0) (label){$\to_4 \; f\enspace.$};
    \end{tikzpicture}
  \end{center}
  Apart from the production incorporated at the root of~$p_{\text f}$,
  no other production of~$G$ targets~$q_{\text{\upshape f}} $ directly, so 
  no more productions are added to~$P_2$.
 
  Finally, we define the WTGh~$G_2 = (Q \cup \{f\} \cup \{\bot\}, \Sigma, 
  F_2, P_2, \mathord{\wt_2})$ with~$F_2(f) = F(q_{\text{\upshape f}} )=1$,
  $F_2(q_{\text{\upshape f}} ) = F_2(q_0) = F_2(\bar{q}) = F_2(\bot) = 0$,
  and~$P_2 = P \cup\{p^1_1,p^1_2\} \cup \{p^2_1, p^2_2\}$.
\end{example}

It remains to show that the existence of a decomposition~$\sem G =
\sem{G_1} + \sem{G_2}$ as in Lemma~\ref{lm: decomp} implies the
nonregularity of~$\sem G$.  For this, we employ the following idea.
Consider a ranked alphabet~$\Sigma$ containing a letter~$\sigma$ of
rank~$2$, a WTA~$G' = (Q, \Sigma, F, P, \mathord{\wt})$ over~$\Sigma$
(which exemplifies~$G_2$), and a sequence~$t_0, t_1, t_2, \dotsc \in
T_\Sigma$ of pairwise distinct trees.  At this point, we assume
that~$P$ contains all possible productions, but we may have~$\wt(p) =
0$ for~$p \in P$.  Using the initial algebra
semantics~\cite{fulvog09}, we can find a matrix representation for the
weights assigned by~$G'$ to trees of the form~$\sigma(t_i, t_j)$ as 
follows.  We enumerate the states~$Q = \{\seq q1n\}$ and for every~$i
\in \N$ define a (column) vector~$\nu_i \in \N^n$ by~$(\nu_i)_k =
\wt_{G'}^{q_k}(t_i)$ for~$k \in [n]$.  Furthermore, we define a matrix~$N
\in \N^{n \times n}$ by~$N_{kh} = \sum_{q \in Q} F(q) \cdot 
\wt(\sigma(q_k, q_h) \to q)$ for $k,h \in [n]$.
Then~$\sem {G'}(\sigma(t_i, t_j)) = \nu_i^{\text T} N \nu_j$ for all~$i,
j \in \N$, where $\nu_i^{\text T}$~is the transpose of~$\nu_i$.

We employ this matrix representation to show that the sum of~$\sem {G'}$
and the (nonregular) characteristic function~$1_L$ of the tree
language~$L = \{\sigma(t_i, t_i) \mid i \in \N\}$ is not regular.  We
proceed by contradiction and assume that~$\sem {G'} + 1_L$ is regular.
Thus we can find an analogous matrix representation using a matrix~$N'$
and vectors~$\nu'_i$ for~$\sem {G'} + 1_L$.  Since the trees~$t_0, t_1,
t_2, \dotsc$ are pairwise distinct, we can write
\[ \bigl(\sem {G'} + 1_L \bigr) \bigl(\sigma(t_i, t_j) \bigr)
  = (\nu'_i)^{\text T} N' \nu'_j = \sem {G'}\bigl(\sigma(t_i, t_j) \bigr)
  + \delta_{ij} = \nu_i^{\text T} N \nu_j + \delta_{ij} \]
for all~$i,j \in \N$, where $\delta_{ij}$~denotes the
\textsc{Kronecker} delta.  The vectors $\nu'_i$~and~$\nu_i$ contain
nonnegative integers, so we may consider the concatenated
vectors~$\langle \nu'_i, \nu_i\rangle$ as vectors of~$\mathbb{Q}^m$
where~$m \in \N$ is the sum of number of states of~$G'$ and of the 
WTA we assumed recognizes~$\sem{G'}+1_L$.  Since~$\mathbb{Q}^m$ 
is a finite dimensional
$\mathbb{Q}$"~vector space, the $\mathbb{Q}$"~vector space spanned by 
the family~$(\langle \nu'_i, \nu_i\rangle)_{i \in \N}$ is also finite
dimensional.  We may thus select a finite generating set
from~$(\langle \nu'_i, \nu_i \rangle)_{i \in \N}$.  For simplicity, we
assume that~$\langle \nu'_1, \nu_1\rangle, \dotsc, \langle \nu'_K,
\nu_K\rangle$ form such a generating set.  Thus there exist~$\seq a1K
\in \mathbb{Q}$ with~$\langle \nu'_{K+1}, \nu_{K+1} \rangle = \sum_{i
  \in [K]} a_i \langle \nu'_i, \nu_i \rangle$.  Applying the usual
distributivity laws for matrix multiplication, we reach a
contradiction as follows.
\begin{linenomath*}
  \begin{align*}
    \bigl(\sem {G'} + 1_L \bigr) \bigl(\sigma(t_{K+1}, t_{K+1}) \bigr)
    &=\, (\nu'_{K+1})^{\text T} N' \nu'_{K+1} \;\!= \sum_{i \in [K]} a_i
      (\nu'_i)^{\text T} N' \nu'_{K+1} \\*
    &= \sum_{i \in [K]} a_i \nu_i^{\text
      T} N \nu_{K+1} = \nu_{K+1}^{\text T} N \nu_{K+1} = \sem {G'}
      \bigl(\sigma(t_{K+1}, t_{K+1}) \bigr)
  \end{align*}
\end{linenomath*}

For the general case, we do not want to assume that $\sem{G_2}$~is
regular, so we cannot assume to have a matrix representation as we had
for~$\sem {G'}$ above.  In order to make our idea work, we identify a set
of trees for which the behavior of~$\sem{G_1} + \sem{G_2}$ resembles
that of~$\sem {G'} +1_L$; more precisely, we construct a context~$C$ and
a sequence~$t_0, t_1, t_2, \dotsc$ of pairwise distinct trees such
that~$(\sem{G_1} + \sem{G_2})(C(t_i, t_j)) = \nu_i^{(1)} N \nu_j^{(2)}
+ \delta_{ij}\mu_i$ for all~$i,j \in \N$ and
additionally,~$\mu_i > 0$ for all~$i \in \N$.
Unfortunately, working with a $2$"~context~$C$ may be insufficient if
$G_1$~uses constraints of the form~$\{v = v', v' = v''\}$, where more
than two positions are constrained to be pairwise equivalent.
Therefore, we have to consider more general $n$"~contexts~$C$ and then
identify a sequence of trees such that the equation above is satisfied
on~$C(t_i, t_j, t_j, \dotsc, t_j)$.

We are now ready for the final theorem. For this, we will use the following version of
\textsc{Ramsey}'s theorem~\cite{ramsey1930}.  For a set~$X$, we denote
by~$\binom X2$ the set of all subsets of~$X$ of size~$2$.

\begin{theorem}
  Let~$k \geq 1$ be an integer and~$f \colon {\binom{\N}2} \to [k]$
  a mapping.  There exists an infinite subset~$E
  \subseteq \N$ such that~$f|_{\binom E2}\equiv i$
  for some~$i\in [k]$.
\end{theorem}


\begin{theorem}
  \label{thm: final}
  Let $G = (Q \cup \{\bot\}, \Sigma, F, P, \mathord{\wt})$ be a trim
  WTGh.  If $G$~satisfies the \ldpp{}, then $\sem G$~is not regular.
\end{theorem}

\begin{proof}
  By Lemma~\ref{lm: decomp} there exist two trim WTGh $G_1 = (Q_1 \cup
  \{\bot\}, \Sigma, F_1, P_1, \mathord{\wt_1})$ and $G_2 = (Q_2 \cup
  \{\bot\}, \Sigma, F_2, P_2, \mathord{\wt_2})$ with~$\sem G(t) =
  \sem{G_1}(t) + \sem{G_2}(t)$ for all~$t \in T_\Sigma$.
  Additionally, there exists~$f\in Q_1$ with~$F_1(f) \neq 0$ 
  and~$F_1(q) = 0$ for all~$q \in Q_1 \setminus
  \{f\}$ and there exists exactly one production~$p_{\text f} = (\ell_{\text f}
  \stackrel {E_{\text f}}\longrightarrow f) \in P_1$ whose target state
  is~$f$.  Finally, for this production~$p_{\text f}$ there
  exists~$(u, v) \in E_{\text f}$ with~$\ell_{\text f}(u) \neq \ell_{\text f}(v) = 
  \bot$ and an infinite sequence~$t_0, t_1, t_2, \dotsc \in T_\Sigma$ of pairwise
  distinct trees with~$D_{G_1}^{\ell_{\text f}(u)}(t_i) \neq \emptyset$ for all~$i
  \in \N$.
	
  Let~$t \in T_\Sigma$ be such that~$D_{G_1}^{f}(t) \neq
  \emptyset$, and let~$\seq w1r$ be an enumeration of all positions
  that are equality-constrained to~$u$ via~$E_{\text f}$, where we assume
  that~$w_1 = u$.  We define a context~$C = t[\Box]_{w_1} \dotsm 
  [\Box]_{w_r}$.  Then~$\sem{G_1}(C(t_i, t_j, t_j, \dotsc, t_j)) > 0$
  if and only if~$i = j$. 
	
 Let us establish some additional notations. Let~$k, h \in \N$ and assume there is~$q \in Q_2$
  with~$F_2(q) \neq 0$ and~$d = (p_1,
  w_1) \dotsm (p_m, w_m) \in D_{G_2}^q(C(t_k, t_h,t_h, \dotsc, t_h))$.
  Let~$p_i = \ell_i \stackrel{E_i}\longrightarrow q_i$ for every~$i
  \in \nve{m}$, and for a set~$X \subseteq \pos(C(t_k,t_h, t_h,
  \dotsc, t_h))$, we let~$i_1 < \dotsb < i_n$ be such that~$\seq
  w{i_1}{i_n}$ is an enumeration of~$\{\seq w1m\} \cap X$; i.e., all
  positions in~$X$ to which~$d$ applies productions. 
  We set~$d|_X = (p_{i_1}, w_{i_1}) \dotsm (p_{i_n}, w_{i_n})$,
  $\wt_2(d|_X) = \prod_{j \in [n]} \wt_2(p_{i_j})$, and
  $D_{kh} = \{d'|_{\pos(C)} \mid \exists q' \in Q_2 \colon F_2(q')
    \neq 0,\, d' \in D_{G_2}^{q'}(C(t_k, t_h, t_h, \dotsc, t_h)) \}$.
	
  We now employ \textsc{Ramsey}'s theorem in the following way.
  For~$k,h \in \N$ with~$k < h$, we consider the mapping~$\{k, h\}
  \mapsto D_{kh}$.  This mapping has a finite range as every~$D_{kh}$
  is a set of finite words over the alphabet~$P_2 \times \pos(C)$ of
  length at most~$\abs{\pos(C)}$.  Thus, by \textsc{Ramsey}'s theorem,
  we obtain a subsequence~$(t_{i_j})_{j \in \N}$ with~$D_{i_ki_h} =
  D_<$ for all~$k, h \in \N$ and some set~$D_<$.  For simplicity, we
  assume~$D_{kh} = D_<$ for all~$k, h \in \N$ with~$k <
  h$. Similarly, we select a further subsequence and
  assume~$D_{kh} = D_>$ for all~$k, h \in \N$ with $k > h$.
  Finally, the mapping~$k \mapsto D_{kk}$ also has a finite range, so
  by the pigeonhole principle, we may select a further subsequence and
  assume that~$D_{kk} = D_=$ for all~$k \in \N$ and some set~$D_=$.
  In the following, we show that~$D_< = D_> \subseteq D_=$.
	
  For now, we assume~$D_< \neq \emptyset$, let~$(p_1, w_1) \dotsm
  (p_m, w_m) \in D_<$, and let~$p_i = \ell_i
  \stackrel{E_i}\longrightarrow q_i$ for every~$i \in \nve{m}$.  Also,
  we define~$C_{kh} = C(t_k, t_h, t_h, \dotsc, t_h)$, $C_{k\Box} =
  C(t_k, \Box, \Box, \dotsc, \Box)$, and~$C_{\Box h} = C(\Box, t_h,
  t_h, \dotsc, t_h)$ for~$k, h \in \N$.  We show that every constraint
  from every~$E_i$ is satisfied on all~$C_{kh}$ with~$k, h \geq 1$,
  not just for~$k < h$.  More precisely, let~$i \in \nve{m}$, $(u',
  v') \in E_i$, and~$(u,v) = (w_iu', w_iv')$.  We show~$C_{kh}|_u =
  C_{kh}|_v$ for all~$k, h \geq 1$.  Note that by assumption,
  $C_{kh}|_u = C_{kh}|_v$~is true for all~$k, h \in \N$ with~$k < h$.
  We show our statement by a case distinction depending on the
  position of $u$~and~$v$ in relation to the positions~$\seq w1r$.
  \begin{enumerate}
  \item If both $u$~and~$v$ are parallel to~$w_1$, then
    $C_{ij}|_u$~and~$C_{ij}|_v$ do not depend on~$i$.  Thus,
    $C_{0j}|_u = C_{0j}|_v$ for all~$j \geq 1$ implies the statement. 
  \item If $u$~is in prefix-relation with~$w_1$ and $v$~is parallel
    to~$w_1$, then $C_{ij}|_v$ does not depend on~$i$.  If~$u \leq
    w_1$, then by our assumption that~$(t_i)_{i \in \N}$ are pairwise
    distinct, we obtain the contradiction~$C_{02}|_v = C_{02}|_u \neq
    C_{12}|_u = C_{12}|_v$, where~$C_{02}|_v = C_{12}|_v$ should
    hold.  Thus, we have~$w_1 \leq u$ and in particular,
    $C_{ij}|_u$~does not depend on~$j$.  Thus, for all~$i,j \geq 1$ we
    obtain~$C_{ij}|_u = C_{i,i+1}|_u = C_{i,i+1}|_v = C_{0,i+1}|_v =
    C_{0,i+1}|_u = C_{0j}|_u = C_{0j}|_v = C_{ij}|_v$.  If $v$~is in
    prefix-relation with~$w_1$ and $u$~is parallel to~$w_1$, then we
    come to the same conclusion by formally exchanging $u$~and~$v$ in
    this argumentation.
  \item If $u$~and~$v$ are both in prefix-relation with~$w_1$, then
    $u$~and~$v$ being parallel to each other implies $w_1 \leq
    u$~and~$w_1 \leq v$.  In particular, both $u$~and~$v$ are parallel
    to all~$\seq w2m$.  Thus, we obtain, as in the first case,
    that $C_{ij}|_u$~and~$C_{ij}|_v$ do not depend on~$j$ and the
    statement follows from~$C_{i,{i+1}}|_u = C_{i,{i+1}}|_v$ for
    all~$i \in \N$.
  \end{enumerate}
  Let~$k, h \geq 1$ and~$d_C \in D_<$, and let~$q \in Q_2$, $d_{k,k+1}
  \in D_{G_2}^q(C_{k,k+1})$, and~$d_{h-1, h} \in D_{G_2}^q(C_{h-1,h})$
  such that~$d_C = d_{k, k+1}|_{\pos(C)} = d_{h-1,h}|_{\pos(C)}$.
  Then for~$d_k = d_{k,k+1}|_{\pos(C_{k,k+1}) \setminus \pos(C_{\Box,k+1})}$ 
  and~$d_h = d_{h-1,h}|_{\pos(C_{h-1,h})\setminus\pos(C_{h-1,\Box})}$,
  we can reorder~$d = d_k d_h d_C$ to a complete left-most derivation
  of~$G_2$ for~$C_{kh}$, as all equality constraints from~$d_k$ are
  satisfied by the assumption on~$d_{k,k+1}$, all equality constraints
  from~$d_h$ are satisfied by the assumption on~$d_{h-1,h}$, and all
  equality constraints from~$d_C$ are satisfied by our case
  distinction.  Considering the special cases~$k = 2$, $h = 1$, and~$k
  = h = 1$, and the definitions of $D_>$~and~$D_=$, we obtain~$d_C \in
  D_{21} = D_>$ and~$d_C \in D_{11} = D_=$, and hence, $D_< \subseteq
  D_>$~and~$D_< \subseteq D_=$.
 
  The converse inclusion~$D_> \subseteq D_<$ follows with an analogous
  reasoning. 
  In conclusion, we obtain~$D_< = D_> \subseteq D_=$.
  By the reasoning above, the case~$D_< = \emptyset$ we excluded
  earlier is only possible if also~$D_> = \emptyset$, in which case we
  again have~$D_< = D_> \subseteq D_=$.
	
  Let~$\seq d1n$ be an enumeration of~$D_<$, $i \in \nve{n}$, and~$k
  \in \N$.  We define the sets
  \begin{linenomath*}
    \begin{align*}
      D^{(1)}_{i,k}
      &= \bigl\{ d|_{\pos(C_{k,k+1})\setminus\pos(C_{\Box,k+1})} \mid d
        \in D_{G_2}^q(C_{k,k+1}),\, d_i = d|_{\pos(C)},\, q \in Q_2
        \bigr\} \\*
      D^{(2)}_{i,k}
      &= \bigl\{ d|_{\pos(C_{k+1,k})\setminus\pos(C_{k+1,\Box})} \mid d
        \in D_{G_2}^q(C_{k+1,k}),\, d_i = d|_{\pos(C)},\, q \in Q_2 \bigr\}\
    \end{align*}
  \end{linenomath*}
  and the corresponding weights $\nu_{i,k}^{(1)} = \sum_{d \in
    D_{i,k}^{(1)}} \wt_2(d)$~and~$\nu_{i,k}^{(2)} = \sum_{d \in
    D_{i,k}^{(2)}} \wt_2(d)$.  Finally, we let~$q_i$ be the target
  state of the last production in~$d_i$ and define~$\nu_i = F_2(q_i)
  \cdot \wt_2(d_i)$.  Then for all~$k, h \in \N$ we
  have~$\sem{G_2}(C_{kh}) = \sum_{i \in [n]} (\nu_{i,k}^{(1)} \cdot
  \nu_i \cdot \nu_{i,h}^{(2)}) + \delta_{kh} \mu_k$ for nonnegative 
  integer weights~$(\mu_j)_{j \in \N}$, which stem from the fact
  that~$D_= \setminus D_< \neq \emptyset$ may hold.  We arrange the
  weights~$\nu_{i,k}^{(1)}$ into a row vector~$\nu_k^{(1)}$, and the
  weights~$\nu_{i,h}^{(2)}$ into a column vector~$\nu_h^{(2)}$, and
  the weights~$\nu_i$ into a diagonal matrix~$N$ such 
  that~$\sem{G_2}(C_{kh}) = \nu_k^{(1)} N \nu_h^{(2)} +\delta_{kh} \mu_k$.

  Recall that~$\sem{G_1}(C_{kh}) > 0$ if and only if~$k = h$ for
  all~$k, h \in \N$.  We can thus modify the weights~$\mu_k$ to
  obtain~$\sem G(C_{kh}) = \sem{G_2}(C_{kh}) + \sem{G_1}(C_{kh}) =
  \nu_k^{(1)} N \nu_h^{(2)} + \delta_{kh} \mu_k$ with~$\mu_k > 0$ for
  all~$k \in \N$.  If $\sem G$~is regular, we can assume a
  representation~$\sem G(C_{kh}) = g(\kappa_k, \kappa_h, \kappa_h, \dotsc,
    \kappa_h)$ for all~$k,h \in \N$, where $\kappa_h$ is a finite
  vector of weights over~$\N$ where each entry corresponds to the sum
  of all derivations for~$t_h$ to a specific state of a weighted tree
  automaton, and~$g$ is a multilinear map encoding the weights of the
  derivations for~$C(\Box, \Box, \dotsc, \Box)$ depending on the
  specific input states at the $\Box$"~nodes and the target state at
  the root~$\varepsilon$.  We choose~$K$ such that the concatenated
  vectors~$\langle \kappa_1, \nu_1^{(1)} \rangle, \dotsc, \langle
  \kappa_K, \nu_K^{(1)}\rangle$ form a generating set of the
  $\mathbb{Q}$"~vector space spanned by~$(\langle \kappa_i,
  \nu_i^{(1)} \rangle)_{i \in \N}$.  Then there are coefficients~$\seq
  a1K \in \mathbb{Q}$ with~$\kappa_{K+1} = \sum_{i \in [K]} a_i
  \kappa_i$ and~$\nu_{K+1}^{(1)} = \sum_{i \in [K]} a_i \nu_i^{(1)}$.
  Thus, we have
  \begin{linenomath*}
    \begin{align*}
      \nu_{K+1}^{(1)} N \nu_{K+1}^{(2)} + \mu_{K+1}
      = g(\kappa_{K+1}, \kappa_{K+1}, \dotsc,
        \kappa_{K+1}) &=\sum_{i \in [K]} a_i g(\kappa_i, \kappa_{K+1}, 
        \dotsc, \kappa_{K+1}) \\*
      &= \sum_{i \in [K]} a_i \nu_i^{(1)} N \nu_{K+1}^{(2)} =
        \nu_{K+1}^{(1)} N \nu_{K+1}^{(2)}
    \end{align*}
  \end{linenomath*}
  which implies~$\mu_{K+1} = 0$ and thus our desired contradiction.
\end{proof}

%


\bibliography{references.bib}

\clearpage
\appendix
\section*{Appendix}

\thmhom*
\begin{proof}
  We construct a WTGc~$G'$ for~$h_{\sem G}$ in two stages.  First, we
  construct the WTGc
  \[ G'' = (Q \cup \{\bot\}, \Delta \cup \Delta \times P, F'', P'',
    \mathord{\wt''}) \]
  such that for every production~$p = \sigma(\seq q1k) \to q \in P$
  and $h(\sigma) = u = \delta(\seq u1n)$, 
  \[ p'' = \Bigl( \langle \delta,p\rangle(\seq u1n) \llbracket \seq
    q1k \rrbracket \stackrel E\longrightarrow q \Bigr) \in
    P'' \quad \text{with} \quad E = \bigcup_{i \in [k]}
    \pos_{x_i}(u)^2 \]
  where the substitution~$\langle \delta,p\rangle(\seq u1n)\llbracket
  \seq q1k \rrbracket$ replaces for every~$i \in [k]$ only the
  left-most occurrence of~$x_i$ in~$\langle \delta,p\rangle(\seq u1n)$
  by~$q_i$ and all other occurrences by~$\bot$.  Moreover~$\wt''(p'')
  = \wt(p)$.  Additionally, we let~$p''_\delta = \delta(\bot, \dotsc, 
  \bot) \to \bot \in P''$  with~$\wt''(p''_\delta) = 1$ for every~$k
  \in \nat$ and $\delta \in \Delta_k \cup \Delta_k \times P$.  No
  other productions are in~$P''$.  Finally, we let~$F''(q) = F(q)$ for
  all~$q \in Q$ and~$F''(\bot) = 0$. 
  
  We can now delete the annotation.  First we remove all productions
  to~$\bot$ that are labeled with symbols from~$\Delta \times P$.
  Second, we use a deterministic relabeling to remove
  the second components of labels of~$\Delta \times P$.  Thus, we
  overall obtain a WTGc~$G'$ such that~$\sem{G'} = h_{\sem
    G}$. Clearly, we have~$\size(G') \in \mathcal{O} \bigl(
  \size(G) \cdot \size(h) \bigr)$ and~$\he(G) \in \mathcal{O}
  \bigl(\size(h) \bigr)$.
  
  The sole purpose of the annotations is to establish a one-to-one 
  correspondence between the valid derivations of~$G$ and those
  of~$G''$, before evaluating the sums to compute~$h_{\sem G}$.  This
  simplifies the understanding of the correctness of the construction,
  but is otherwise superfluous and may be omitted for efficiency.
\end{proof}


\lmtrim*
\begin{proof}
  First, recall that we may assume~$\wt(p) \neq 0$ for every~$p \in
  P$ because~$\wt_G(d) = 0$ for every derivation~$d$ of~$G$ that
  contains a production~$p$ with~$\wt(p) = 0$.  For the proof, we employ
  a simple reachability algorithm.  For every~$n \in \N$ and~$U
  \subseteq Q$, let
  \begin{linenomath*}
    \begin{align*}
      Q_0 &= \emptyset \quad
      & Q_{n+1}
      &= Q_n \cup \bigcup_{\substack{(\ell \stackrel E\longrightarrow q)
        \in P \\ \ell \in T_\Sigma(Q_n)}} \{q\} \quad
          & \Pi_U
      &= \bigcup_{\substack{(\ell \stackrel E\longrightarrow q)
        \in P \\ \ell \in T_\Sigma(U)}} \bigl\{ (q, q') \in U^2 \mid
      \pos_{q'}(\ell) \neq \emptyset \bigr\} \enspace.
    \end{align*}
  \end{linenomath*}
  Since $Q$~is finite, there exists~$\fixedn$ with~$Q_{\fixedn} =
  Q_{\fixedn+1}$.  Let~$Q' = Q_{\fixedn}$.  A straightforward proof
  shows that~$q \in Q'$ if and only if for some~$t \in T_\Sigma$ there
  exists~$d \in D_G^q(t)$ with~$\wt_G(d) \neq 0$.  To ensure the
  reachability of a final state, we let $\triangleleft$~be the 
  smallest reflexive and transitive relation on~$Q'$ that
  contains~$\Pi_{Q'}$.  Then~$P' = \{ \ell\stackrel
  E\longrightarrow q \in P \mid q\in Q',\, \exists q_{\text
    f} \in Q' \colon F(q_{\text f}) \neq 0,\, q_{\text f}
  \triangleleft q\}$, and the desired WTGh is simply~$G' = (Q \cup \{\bot\}, 
  \Sigma, F, P', \mathord{\wt|_{P'}})$.
\end{proof}

\lmdecomp*
\begin{proof}
  Let $p = (\ell \stackrel E\longrightarrow q) \in P$ be a production
  as in the \ldpp{}.  Since $G$~is trim, there exist a tree~$t' \in
  T_\Sigma$, a final state~$q_{\text f} \in Q$ with~$F(q_{\text f})
  \neq 0$, a derivation~$d = (p_1, w_1) \cdots (p_m, w_m) \in
  D_G^{q_{\text f}}(t')$, and~$i \in [m]$ such that~$p_i = p$.  In
  other words, there is a derivation utilizing production~$p$.  We
  let~$p_j = \ell_j \stackrel{E_j}\longrightarrow q_j$ for every~$j
  \in \nve{m}$, and let~$w_{i_1} > \dotsb > w_{i_k}$ be the sequence
  of prefixes of~$w_i$ among the positions~$\{\seq w1m\}$ in strictly
  descending order with respect to the prefix order.  In particular,
  we have~$w_{i_1} = w_i$ and~$w_{i_k} = \varepsilon$.

  For a position~$w$ and a set~$E'$ of constraints, we define~$wE' =
  \{(wu, wv) \mid (u,v) \in E'\}$.  We want to join the left-hand
  sides of the productions~$\seq p{i_1}{i_k}$ to a new
  production~$\ell_{i_k}[\ell_{i_{k-1}}]_{w_{i_{k-1}}} \dotsm
  [\ell_{i_1}]_{w_{i_1}} \stackrel{E_{\text f}}\longrightarrow
  q_{\text f}$ with~$E_{\text f} = \bigcup_{j \in [k]} w_{i_j}
  E_{i_j}$.  However, we need to ensure that $\seq w{i_1}{i_k}$ do not
  occur in~$E_{\text f}$.  Therefore, we assume that~$p$, $t'$,
  $q_{\text f}$, $d$, and~$i$ above are chosen such that $w_i$~is of
  minimal length among all possible choices.  Then we see as follows
  that $\seq w{i_1}{i_k}$~do not occur in~$E_{\text f}$.

  Let~$(u, v) \in E$ with $\ell(u) \neq \ell(v) = \bot$ and~$t \in
  T_\Sigma$ with~$\he(t) > \he(G)$ and~$D_G^{\ell(u)}(t) \neq
  \emptyset$.  Suppose there exists~$j \in \nve{k}$ such that
  $w_{i_j}$~occurs in~$E_{\text f}$.  Then there exists~$(u', v') \in
  E_{i_{j+1}}$ with~$w_{i_j} = w_{i_{j+1}}u'$.  Then the tree~$t' \sem
  t_{w_iu} |_{w_{i_{j}}}$ shows us that $p_{i_{j+1}}$~is also a
  production as in the \ldpp{}, but $\abs{w_{i_{j+1}}} < 
  \abs{w_i}$, so $w_i$~is not of minimal length.
  
  We define~$G_1 = (Q_1 \cup \{\bot\}, \Sigma, F_1, P_1,
  \mathord{\wt_1})$ as follows.  Let~$f \notin Q \cup \{\bot\}$ be a
  new state.  We set~$Q_1 = Q \cup \{f\}$, $F_1(f) = F(q_{\text f})$,
  and~$F_1(q') = 0$ for all~$q' \in Q$.  For the production~$p_{\text
    f} = (\ell_{i_k}[\ell_{i_{k-1}}]_{w_{i_{k-1}}} \dotsm
  [\ell_{i_1}]_{w_{i_1}} \stackrel{E_{\text f}}\longrightarrow f)$
  with~$E_{\text f} = \bigcup_{j \in [k]} w_{i_j}E_{i_j}$, we
  let~$P_1 = P \cup \{p_{\text f} \}$, $\wt_1(p_{\text f}) =
  \prod_{j \in [k]} \wt(p_{i_j})$, and~$\wt_1(p') = \wt(p')$ for
  all~$p' \in P$.  Then $G_1$~simulates all derivations of~$G$ with
  productions~$\seq p{i_1}{i_k}$ at the positions~$\seq
  w{i_1}{i_k}$, respectively.  For the existence of the infinite
  sequence of trees, let~$(u,v) \in E$ with $\ell(u) \neq \ell(v) =
  \bot$ and~$t \in T_\Sigma$ with~$\he(t) > \he(G)$
  and~$D_G^{\ell(u)}(t) \neq \emptyset$.  By Lemma~\ref{lm:existence
    of subs}, there exists an infinite sequence~$t_0, t_1, t_2, \dotsc
  \in T_\Sigma$ of pairwise distinct trees with~$D_G^{\ell(u)}(t_i)
  \neq \emptyset$ for all~$i \in \N$. Since~$D_G^{\ell(u)}(t_i)
  \subseteq D_{G_1}^{\ell(u)}(t_i)$ for all~$i \in \N$, this is the
  desired sequence.  We conclude the definition of~$G_1$ by noting
  that~$(w_iu, w_iv) \in E_{\text f}$ and that the left-hand 
  side~$\ell_{\text f}$ of~$p_{\text f}$ satisfies~$\ell_{\text
    f}(w_iu) = \ell(u)$.
  
  Next, we construct~$G_2$ such that it simulates all remaining
  derivations of~$G$ in the following sense.  If $d$~is a derivation
  of~$G$ to a state different from~$q_{\text f}$, then it is a
  derivation of~$G_2$.  If $d$~is a derivation of~$G$ to~$q_{\text f}$
  but its last production is not~$p_{i_k}$, then it is simulated
  by a derivation of~$G_2$ to a new state~$f$.  If $d$~is a derivation
  of~$G$ and its last production is~$p_{i_k}$ but the production
  at~$w_{i_{k-1}}$ is not~$p_{i_{k-1}}$, then it again is simulated by
  a derivation of~$G_2$ to~$q_{\text f}$, and so on.  To have a more
  compact definition for~$G_2$, we use the symbol~$\Box$ to denote a
  tree of height~$0$ and a term~$\Box[\ell_{i_k}]_{w_{i_k}} 
  \dotsm [\ell_{i_{j+1}}]_{w_{i_{j+1}}} [\ell']_{w_{i_{j}}}$ 
  for~$j = k$ is to be read as~$\Box[\ell']_{w_{i_{j}}}$.  We let~$f
  \notin Q \cup \{\bot\}$ be a new state and define~$G_2 = (Q_2 \cup
  \{\bot\}, \Sigma, F_2, P_2, \mathord{\wt_2})$ by~$Q_2 = Q \cup
  \{f\}$, $F_2(q_{\text f}) = 0$, $F_2(f) = F(q_{\text f})$,
  and~$F_2(q') = F(q')$ for all~$q' \in Q \setminus\{q_{\text f}\}$.
  Moreover, we let
  \begin{linenomath*}
    \[ P_2 = P \cup \bigcup_{j \in [k]} \Bigl\{
      \Box[\ell_{i_k}]_{w_{i_k}} \dotsm
      [\ell_{i_{j+1}}]_{w_{i_{j+1}}}[\ell']_{w_{i_{j}}}
      \stackrel{E_f}\longrightarrow f \;\Big|\;
      \begin{aligned}[t]
        &p' = (\ell' \stackrel{E'}\longrightarrow q_{i_{j}}) \in P
        \setminus \{p_{i_{j}}\}, \\*
        & E_f = w_{i_{j}} E' \cup \bigcup_{j' = j+1}^k
        w_{i_{j'}}E_{i_{j'}} \Bigr\} \enspace.
      \end{aligned}
    \]
  \end{linenomath*}
  For a production~$p_f = \Box[\ell_{i_k}]_{w_{i_k}} \dotsm
  [\ell_{i_{j+1}}]_{w_{i_{j+1}}}[\ell']_{w_{i_{j}}}
  \stackrel{E_f}\longrightarrow f$ constructed from~$p'$ as above we 
  let~$\wt_2(p_f) = \wt(p') \cdot \prod_{j' = j+1}^k \wt(p_{i_{j'}})$
  and for every~$p' \in P$ we let~$\wt_2(p') = \wt(p')$.  Then we have
  $\sem G(t) = \sem{G_1}(t) + \sem{G_2}(t)$ for every~$t \in
  T_\Sigma$.  Note that trimming $G_1$~and~$G_2$ will not remove any
  of the newly added productions under the assumption that $G$~is
  trim.
\end{proof}

\end{document}